\documentclass[12pt]{article}
\usepackage[utf8]{inputenc} 
\usepackage{braket}
\usepackage{amsmath}
\usepackage{latexsym}
\usepackage{dsfont}
\usepackage{physics}
\usepackage{dutchcal}
\usepackage[normalem]{ulem}
\usepackage{hhline}
\usepackage{amsfonts}
\usepackage[table]{xcolor}
\usepackage[english]{babel}
\usepackage{graphicx}
\usepackage{pgfplots}
\pgfplotsset{width=9cm,compat=1.9}
\graphicspath{ {images/} }
\usepackage{hyperref}
\hypersetup{
    colorlinks=true,
    linkcolor=blue,
    filecolor=magenta,      
    urlcolor=cyan,
}
\usepackage[inline]{enumitem}
\usepackage{amsthm}
\usepackage{tikz-feynman}
\usepackage{stackengine}
\newtheorem{theorem}{Theorem}[section]

\newtheorem{lemma}[theorem]{Lemma}

\usepackage{geometry}
\newcommand{\be}{\begin{equation}}
\newcommand{\ee}{\end{equation}}
\newcommand{\bes}{\begin{subequations}}
\newcommand{\ese}{\end{subequations}}
\newcommand{\beg}{\begin{gather}}
\newcommand{\eg}{\end{gather}}
\newcommand{\ben}{\begin{enumerate}}
\newcommand{\een}{\end{enumerate}}
\def\beqn{\begin{eqnarray}}
\def\eeqn{\end{eqnarray}}

\begin{document}

\begin{flushright}
ITEP-TH-07 /19 
\end{flushright}
\begin{center}
{\large{\bf The Dynamical equations for ${\mathfrak{gl}}(n\vert m)$
 }
}\\
\vspace{10mm} {E. Dotsenko$^{1,2,3}$  %{ M. Olshanetsky}$^{\,\flat\,\S}$
 %{ A. Zotov}$^{\,\diamondsuit\,\flat}$}
 \\ \vspace{7mm}
 \vspace{2mm} $1$ - {\sf Institute for Theoretical and Experimental Physics, B.Cheremushkinskaya 25, 
  Moscow, 117218, Russia}\\
 %\vspace{2mm}$^\natural$ - {\sf International Laboratory for Mirror Symmetry and Automorphic Forms,\\
  %Mathematics Department of NRU HSE,
 %Usacheva str. 6,  Moscow, 119048, Russia}\\
 \vspace{2mm} $2$ - {\sf Moscow Institute of physics and Technology,y, Inststitutskii per. 9, Dolgoprudny, Moscow region, 141700
 Russia}\\
 \vspace{2mm}$3$ - {\sf Skolkovo Institute of Science and Technology, 143026 Moscow, Russia \\
}
 \vspace{4mm}
 {\footnotesize email:  edotsenko95@gmail.com\\}
 \vspace{4mm}}
%\textbf{Abstract}
 \end{center}
 
\begin{abstract}
    In this note we propose a compatible set of equations which commutes with the Knizhnik-Zamolodchikov equations based on the $\mathfrak{{gl}}(n\vert m)$ symmetry algebra and establish the Matsuo-Cherednik correspondence in this context.
\end{abstract}
%\paragraph{Keywords:} Knizhnik-Zamolodchikov equiations, bispectrality, quantum-quantum correspondence, Integrable systems\par \emph{This article is registered under preprint
%number: arXiv:1904.00006}
\section{Introduction}

There is a relation between the solutions of the (quantum) Knizhnik-Zamolodchikov equations ((q)KZ for short) and the spectral problem of quantum systems Calogero-Moser type, proposed in \cite{M} and later extended in \cite{Ch}, \cite{Mi}. The recent discussion \cite{ZZ}, \cite{GZZ} originated in the quantum-classical correspondence \cite{NRS}, \cite{GoZZ}. \par

It was shown in \cite{FMTV} and generalized in \cite{EV}, that for twisted KZ equations, based on simple finite-dimensional Lie algebra $\mathfrak{g}$, there is a compatible set of equations with respect to twist parameters, which commutes with KZ. This set of equations was called Dynamical equations. The joint system could be viewed as a 'vector' analogue of the bispectral problem of \cite{DG}. In fact, they are related through the Matsuo-Cherednik map (MC for brevity): if one has a solution of a joint system of KZ and Dynamical equations, then one can cook-up the solution of the bispectral problem for the corresponding system of Calogero-Moser type. For the proof one can consult with \cite{TL} (section 7), \cite{MSt} and \cite{BMO}. It is worth to mention that in \cite{OS} the KZ and Dynamical equations were defined for a wide class of algebras and it will be very interesting to study the MC map in this generality.\par

  Integrable systems with $\mathbb{Z}/2\mathbb{Z}$ graded symmetry algebras (mostly of type $A$) were studied in many papers, for example, one can consult with \cite{BR}. Also recently superchains were embedded in the context of the Bethe/Gauge correspondence in \cite{N}. And, maybe, it might be interesting to derive the Dynamical equations, proposed in this note using the approach of \cite{OS}. \par

The outline of the work is the following: in section 2 we will prove the compatibility of the proposed Dynamical equations and in section 3 we will establish the MC map. The q-difference generalization of the proposed equations,  as well as issues about the monodromy, will be discussed elsewhere \cite{E}. 

\paragraph{Acknowledgements}\par
The author has benefited from the discussions with B. Feigin, A. Zabrodin, A. Zotov and N. Slavnov. Also, the author is grateful to P. Etingof for the reference \cite{NatG}.  
This work was supported by the RFBR -18-02-01081 grant.
\section{The Dynamical equations}
Let $\overrightarrow{z}\in \mathbb{C}^{k}\backslash\amalg_{i<j} z_{i}= z_{j}$  and $\overrightarrow{\lambda}\in \mathbb{C}^{n+m}\backslash \amalg_{i<j} \lambda_{i}= \lambda_{j}$. Let $V = \bigotimes^{k}_{i=1}\mathbb{C}^{n\vert m}$ be the tensor product of the vector representations of $\mathfrak{gl}(n\vert m)$.
And let $\vert\Psi\rangle : (\overrightarrow{z},\overrightarrow{\lambda}) \to V$ be the flat section, namely
\begin{equation}
    \left(\kappa\partial_{z_i} - \sum_{c}\lambda_c e^{(i)}_{cc} - \sum_{j,\neq i} \frac{P_{ij}}{z_{i} - z_{j}}\right)\vert\Psi\rangle = 0, \label{sKZ}
\end{equation}
where graded permutation $P_{ij}= \sum_{a,b} (-1)^{p(b)}e^{(i)}_{ab}e^{(j)}_{ba}$ and $p(a)$ is a parity function, defined in the Appendix. $e^{(j)}_{ab}$ is the generator of $\mathfrak{gl}(n\vert m)$ which non-trivially acts only in the $i^{th}$ tensor component of $V$ as matrix unit in some basis (see Appendix).
\begin{theorem}
The following system of equations is compatible and commutes with $\eqref{sKZ}$
\begin{equation}
     \left(\kappa\partial_{a} - \sum_{j}z_{j}e^{(j)}_{aa} - \sum_{b,\neq a}(-1)^{p(b)}\frac{E_{ab}E_{ba} - E_{aa}}{\lambda_a - \lambda_b}\right)\vert\Psi\rangle = 0, \label{sDyn}
\end{equation}
where $E_{ab} = \sum^{k}_{j = 1}e^{(j)}_{ab}$.
\end{theorem}
\begin{proof}
Let us prove first, that $\eqref{sDyn}$ is a compatible system. For showing this one has to check that the following commutator vanishes 
\begin{equation}
\begin{gathered}
     \Bigg[\kappa\partial_{a} - \sum_{j}z_{j}e^{(j)}_{aa} - \sum_{c,\neq a}(-1)^{p(c)}\frac{E_{ac}E_{ca} - E_{aa}}{\lambda_a - \lambda_c},  \kappa\partial_{b} - \sum_{l}z_{l}e^{(l)}_{bb} - \sum_{d,\neq b}(-1)^{p(d)}\frac{E_{bd}E_{db} - E_{bb}}{\lambda_b - \lambda_d} \Bigg] = 0.
     \end{gathered}
\end{equation}
There are three types of terms 
\begin{enumerate}
 \item[(1)] 
 \begin{equation}
 \begin{gathered}
     \Bigg[\partial_{a}, \sum_{d,\neq b}(-1)^{p(d)}\frac{E_{bd}E_{db} - E_{bb}}{\lambda_b - \lambda_d} \Bigg] +  \Bigg[\sum_{c,\neq a}(-1)^{p(c)}\frac{E_{ac}E_{ca} - E_{aa}}{\lambda_a - \lambda_c},\partial_{b}\Bigg] = \\ = (-1)^{p(a)}\frac{E_{ba}E_{ab} - E_{bb}}{(\lambda_b - \lambda_a)^2} - (-1)^{p(b)}\frac{E_{ab}E_{ba} - E_{aa}}{(\lambda_a - \lambda_b)^2} = 0.
     \end{gathered}
 \end{equation}
 The last equation holds due to relations $\eqref{definingglnm}$.
\item[(2)] 
\begin{equation}
    \begin{gathered}
    \Bigg[ \sum_{j}z_{j}e^{(j)}_{aa}, \sum_{d,\neq b}(-1)^{p(d)}\frac{E_{bd}E_{db} - E_{bb}}{\lambda_b - \lambda_d} \Bigg] +\Bigg[ \sum_{c,\neq a}(-1)^{p(c)}\frac{E_{ac}E_{ca} - E_{aa}}{\lambda_a - \lambda_c},  \sum_{l}z_{l}e^{(l)}_{bb}\Bigg] = \\ = \frac{(-1)^{p(a)}}{\lambda_{ba}}\left(\sum_{j}z_{j}(E_{ba}e^{(j)}_{ab}- e^{(j)}_{ba}E_{ab})\right) + \frac{(-1)^{p(b)}}{\lambda_{ab}}\left(\sum_{j}z_{j}(e^{(j)}_{ab}E_{ba}- E_{ab}e^{(j)}_{ba})\right) = 0
    \end{gathered}
\end{equation}
The last equation is true due to $\eqref{superotimes}$
\item[(3)] 
\be
 \begin{gathered}
 \Bigg[\sum_{c,\neq a}(-1)^{p(c)}\frac{E_{ac}E_{ca} - E_{aa}}{\lambda_a - \lambda_c}, \sum_{d,\neq b}(-1)^{p(d)}\frac{E_{bd}E_{db} - E_{bb}}{\lambda_b - \lambda_d} \Bigg] = \\
 = \sum_{c,\neq a,b}\Bigg(\Bigg[ (-1)^{p(b)}\frac{E_{ab}E_{ba}}{\lambda_{ab}}, (-1)^{p(c)}\frac{E_{bc}E_{cb}}{\lambda_{bc}} \Bigg] + \Bigg[\frac{E_{ac}E_{ca}}{\lambda_{ac}}, \frac{E_{bc}E_{cb}}{\lambda_{bc}} \Bigg] + \\ + \Bigg[ (-1)^{p(c)}\frac{E_{ac}E_{ca}}{\lambda_{ac}}, (-1)^{p(a)}\frac{E_{ba}E_{ab}}{\lambda_{ba}} \Bigg]\Bigg)  +  (-1)^{p(a) + p(b)}\Bigg[\frac{E_{ab}E_{ba}}{\lambda_{ab}},\frac{E_{ba}E_{ab}}{\lambda_{ba}}\Bigg] \label{manycomm}
 \end{gathered}
\ee
A short calculation shows that the last commutator vanishes. Let us focus on the middle row in the above expression, which is more complicated. To prove that it is zero one has to consider the following four cases \par

     \paragraph{ p(a) = p(b) = 0, p(c)= 1} \par
    \be
\begin{gathered}
-\Bigg[\frac{E_{ab}E_{ba}}{\lambda_{ab}},\frac{E_{bc}E_{cb}}{\lambda_{bc}} \Bigg] + \Bigg[\frac{E_{ac}E_{ca}}{\lambda_{ac}}, \frac{E_{bc}E_{cb}}{\lambda_{bc}} \Bigg] - \Bigg[\frac{E_{ac}E_{ca}}{\lambda_{ac}},\frac{E_{ba}E_{ab}}{\lambda_{ba}} \Bigg] = \\ =  -\frac{E_{ac}E_{ba}E_{cb} - E_{bc}E_{ab}E_{ca}}{\lambda_{ab}\lambda_{bc}} + \frac{E_{ac}E_{ba}E_{cb} - E_{bc}E_{ab}E_{ca}}{\lambda_{ac}\lambda_{bc}} - \frac{E_{ac}E_{ba}E_{cb} - E_{bc}E_{ab}E_{ca}}{\lambda_{ac}\lambda_{ba}} = 0.
\end{gathered}
     \ee
     \par
     \paragraph{ p(a) = p(c) = 0, p(b)= 1} 
     \be
     \begin{gathered}
     -\Bigg[\frac{E_{ab}E_{ba}}{\lambda_{ab}},\frac{E_{bc}E_{cb}}{\lambda_{bc}} \Bigg] + \Bigg[\frac{E_{ac}E_{ca}}{\lambda_{ac}}, \frac{E_{bc}E_{cb}}{\lambda_{bc}} \Bigg] + \Bigg[\frac{E_{ac}E_{ca}}{\lambda_{ac}},\frac{E_{ba}E_{ab}}{\lambda_{ba}} \Bigg] = \\ = \frac{E_{ac}E_{ba}E_{cb} - E_{bc}E_{ab}E_{ca}}{\lambda_{ab}\lambda_{bc}} - \frac{E_{ac}E_{ba}E_{cb} - E_{bc}E_{ab}E_{ca}}{\lambda_{ac}\lambda_{bc}} + \frac{E_{ac}E_{ba}E_{cb} - E_{bc}E_{ab}E_{ca}}{\lambda_{ac}\lambda_{ba}} = 0.
     \end{gathered}
     \ee
     \paragraph{ p(b) = p(c) = 1, p(a)= 0}
     \be
     \begin{gathered}
     \Bigg[\frac{E_{ab}E_{ba}}{\lambda_{ab}},\frac{E_{bc}E_{cb}}{\lambda_{bc}} \Bigg] + \Bigg[\frac{E_{ac}E_{ca}}{\lambda_{ac}}, \frac{E_{bc}E_{cb}}{\lambda_{bc}} \Bigg] - \Bigg[\frac{E_{ac}E_{ca}}{\lambda_{ac}},\frac{E_{ba}E_{ab}}{\lambda_{ba}} \Bigg] = \\ = \frac{E_{ac}E_{ba}E_{cb} - E_{bc}E_{ab}E_{ba}}{\lambda_{ab}\lambda_{bc}} - \frac{E_{ac}E_{ba}E_{cb} - E_{bc}E_{ab}E_{ca}}{\lambda_{ac}\lambda_{bc}} + \frac{E_{ac}E_{ba}E_{cb} - E_{bc}E_{ab}E_{ca}}{\lambda_{ac}\lambda_{ba}} = 0.
     \end{gathered}
     \ee
     \paragraph{ p(a) = p(b) = 1, p(c)= 0}
     \be
     \begin{gathered}
     -\Bigg[\frac{E_{ab}E_{ba}}{\lambda_{ab}},\frac{E_{bc}E_{cb}}{\lambda_{bc}} \Bigg] + \Bigg[\frac{E_{ac}E_{ca}}{\lambda_{ac}}, \frac{E_{bc}E_{cb}}{\lambda_{bc}} \Bigg] - \Bigg[\frac{E_{ac}E_{ca}}{\lambda_{ac}},\frac{E_{ba}E_{ab}}{\lambda_{ba}} \Bigg] = \\ = -\frac{E_{ac}E_{ba}E_{cb} - E_{bc}E_{ab}E_{ca}}{\lambda_{ab}\lambda_{bc}} + \frac{E_{ac}E_{ba}E_{cb} - E_{bc}E_{ab}E_{ca}}{\lambda_{ac}\lambda_{bc}} - \frac{E_{ac}E_{ba}E_{cb} - E_{bc}E_{ab}E_{ca}}{\lambda_{ac}\lambda_{ba}} = 0.
     \end{gathered}
     \ee
\end{enumerate}
The two cases where parity of $a,b$ and $c$ are equal are not considered because the three corresponding generators of the $\mathfrak{gl}(n\vert m)$ are bosonic. \par
From the above computation, one sees that the system $\eqref{sDyn}$ is compatible and the first part of the theorem is proven. \par
Now let us show that systems $\eqref{sKZ}$ and $\eqref{sDyn}$ are compatible. Again one needs to check that the following commutator vanishes
\be
\left[ \kappa\partial_{z_i} - \sum_{c}\lambda_c e^{(i)}_{cc} - \sum_{j,\neq i} \frac{P_{ij}}{z_{i} - z_{j}}, \kappa\partial_{a} - \sum_{j}z_{j}e^{(j)}_{aa} - \sum_{b,\neq a}(-1)^{p(b)}\frac{E_{ab}E_{ba} - E_{aa}}{\lambda_a - \lambda_b} \right] = 0.
\ee
Let us consider the most complicated parts of the above commutator
\begin{enumerate}
    \item \be
\begin{gathered}
\left[\sum_{c}\lambda_c e^{(i)}_{cc}, \sum_{b,\neq a}(-1)^{p(b)}\frac{E_{ab}E_{ba} - E_{aa}}{\lambda_{ab}} \right] + \left[\sum_{j,\neq i} \frac{P_{ij}}{z_{ij}}, \sum_{j}z_{j}e^{(j)}_{aa} \right] = \\ = \sum_{b,\neq a}(-1)^{p(b)}\left(e^{(i)}_{ab}E_{ba} - E_{ab}e^{(i)}_{ba} - e^{(i)}_{ab}E_{ba} + E_{ab}e^{(i)}_{ba}\right) = 0.
\end{gathered}
\ee
\item \be
\left[ \sum_{j,\neq i} \frac{P_{ij}}{z_{i} - z_{j}}, \sum_{b,\neq a}(-1)^{p(b)}\frac{E_{ab}E_{ba} - E_{aa}}{\lambda_a - \lambda_b} \right] = 0.
\ee
\end{enumerate}

Indeed, one has to check that 
\be
\left[ P_{ij},E_{ab}E_{ba}\right] = 0. \label{permcomm}
\ee
A simple calculation shows that the above equality is correct. 
\end{proof}
\section{ The Matsuo-Cherednik map}
In this section let us fix $k = n+m$ in the definition of $V$. 
\begin{theorem}
The following covectors, constructed in \cite{GZZ} provide the Matsuo-Cherednik map for Dynamical equations $\eqref{sDyn}$
\begin{subequations}
    \begin{gather}
    \langle \Omega^{0}\vert = \sum_{\sigma\in S_{n+m}}\langle e_{1}\otimes \ldots \otimes e_{n+m}\vert P_{\sigma}, \label{proj+} \\
    \langle \Omega^{1}\vert  = \sum_{\sigma\in S_{n+m}}\langle e_{1}\otimes \ldots \otimes e_{n+m}\vert (-1)^{sgn(\sigma)}P_{\sigma}, \label{proj-}
    \end{gather}
\end{subequations}
where $P_{\sigma} = P_{s_{i_1}}\ldots P_{s_{i_l}}$, where $P_{s_{i_j}}$ is a permutation, that corresponds to the reflection $s_{i_j}$ and $\sigma  = s_{i_1}s_{i_2}\ldots s_{i_{l}}$ is some decomposition into transpositions. Because $P_{s_{i_j}}$ satisfy braid relation the element $P_{\sigma}$ is correctly defined.
\end{theorem}
The appearance of $\eqref{proj-}$ is a property of $V[1]$ on which we are projecting the Dynamical equation. The  $V[1]\subset V$ is a subspace, such that $E_{aa}\vert_{V[1]} = 1$ for all $a \leq n+m$. Such condition is consistent with equations $\eqref{sKZ}$ and $\eqref{sDyn}$ because the Cartan generators commute with them. And presumably one can not consider the other weight subspaces different to $V[1]$, because the symmetry between $z's$ and $\lambda's$ is broken. Now let us prove a simple and useful 
\begin{lemma}
The covectors $\eqref{proj+}$ and $\eqref{proj-}$ are eigenvectors of the following operators
\begin{equation}
        \langle \Omega^{i}\vert (E_{ab}E_{ba}-E_{aa}) = (-1)^{p(b) + i}\langle \Omega^{i}\vert, \label{eigenp}
\end{equation}
where $i = 0,1$.
\end{lemma}

\begin{proof}
Instead of considering the above covectors one can consider the vectors $\vert \Omega^{i}\rangle$ and prove for them the analogue properties $\eqref{eigenp}$. A simple computation shows that
\be
\left(E_{ab}E_{ba} - E_{aa}\right) \vert e_{1}\otimes \ldots\otimes e_{n+m}\rangle = (-1)^{p(b)}P_{ab}\vert e_{1}\otimes \ldots\otimes e_{n+m} \label{simple} \rangle.
\ee
Then, symmetrize or skew symmetrize the r.h.s. of $\eqref{simple}$ one obtains $\eqref{eigenp}$ immediately.
\end{proof}
Now we are ready to prove the following 
\begin{theorem}
Let $\vert\Psi\rangle$ be the solution of joint system $\eqref{sKZ}$, $\eqref{sDyn}$ with values in $V[1]$, then
\begin{subequations}
\begin{gather}
\left( \kappa^2\sum^{m+n}_{a = 1} +\sum_{a\neq b}\frac{(-1)^i\kappa -1}{(z_{ab})^2}\right)\langle\Omega^{i}\vert\Psi\rangle = \left(\sum^{n+m}_{i = 1}\lambda^{2}_{i}\right)\langle\Omega^{i}\vert\Psi\rangle, \label{GrZabZot} \\
   \left( \kappa^2\sum^{m+n}_{a = 1} +\sum_{a\neq b}\frac{(-1)^i\kappa -1}{(\lambda_{ab})^2}\right)\langle\Omega^{i}\vert\Psi\rangle = \left(\sum^{n+m}_{i = 1}z^{2}_{i}\right)\langle\Omega^{i}\vert\Psi\rangle,
   \end{gather}
\end{subequations}
where $i = 0,1$
\end{theorem}
\begin{proof}
The $\eqref{GrZabZot}$ was proven in \cite{GZZ}, so we are to prove the second. Let $D_a$ be the $a^{th}$ dynamical operator. And let us consider the following sum
\begin{equation}
\begin{gathered}
     \langle \Omega ^{i}\vert \sum^n_{a = 1}D^2_a\vert\Psi\rangle = \langle \Omega \vert \Bigg(\sum^{n}_{a = 1}\kappa\frac{\partial^2}{\partial\lambda^2_a} - \{\sum^{n}_{i,a = 1}z_ie^{(i)}_{aa},\sum_{b\neq a}(-1)^{p(b)}\frac{E_{ab}E_{ba} - E_{aa}}{\lambda_a - \lambda_b}\}- \\ - \left( \sum_{i} z_ie^{(i)}_{aa}\right)^2 - \sum_{b\neq c\neq a}\frac{(-1)^{p(b)+p(c)}(E_{ab}E_{ba} - E_{aa})(E_{ac}E_{ca} - E_{aa})}{(\lambda_a - \lambda_b)(\lambda_a - \lambda_c)} + \\ + \sum_{b\neq a}\frac{\kappa(-1)^{p(b)}(E_{ab}E_{ba} - E_{aa}) +(E_{ab}E_{ba} - E_{aa})^2}{(\lambda_a - \lambda_b)^2}\Bigg)\vert\Psi\rangle, \label{sumofsquares}
\end{gathered}
\end{equation}
in passing from l.h.s. to r.h.s. of $\eqref{sumofsquares}$ the $\eqref{sDyn}$ were used. After applying $\eqref{eigenp}$ and the identity $\sum_{b\neq c\neq a}\frac{1}{(\lambda_a - \lambda_b)(\lambda_a - \lambda_c)} = 0$ the above expression simplifies to
\begin{equation}
\langle \Omega ^{i}\vert (\sum^{n+m}_{a = 1}\kappa^2\frac{\partial^2}{\partial\lambda^2_a} - \{\sum^{n+m}_{i,a = 1}z_ie^{(i)}_{aa},\sum_{b\neq a}(-1)^{p(b)}\frac{E_{ab}E_{ba} - E_{aa}}{\lambda_a - \lambda_b}\}-  \sum_{i} z^2_i + \sum_{b\neq a}\frac{\kappa(-1)^{i} - 1}{(\lambda_a - \lambda_b)^2})\vert\Psi\rangle = 0.
\end{equation}
To prove the theorem one has to show, that
\be
    \langle \Omega^{i} \vert\left( \{\sum^{n+m}_{i,a = 1}z_ie^{(i)}_{aa},\sum_{b\neq a}(-1)^{p(b)}\frac{E_{ab}E_{ba} - E_{aa}}{\lambda_a - \lambda_b}\}\right)\vert\Psi\rangle = 0.
\ee
In fact, one has to prove, that for any distinct $a$ and $b$
\begin{equation}
    \langle \Omega^{i} \vert\left( \{\sum^{n+m}_{i = 1}z_ie^{(i)}_{aa},(-1)^{p(b)}(E_{ab}E_{ba} - E_{aa})\}-\{\sum^{n+m}_{i = 1}z_ie^{(i)}_{bb},(-1)^{p(a)}(E_{ba}E_{ab} - E_{bb})\}\right ) = 0.
\end{equation}
After opening the brackets one obtains the following
\begin{equation}
    \begin{gathered}
         \langle \Omega^{i}\vert\left((\sum^{n}_{i = 1}z_ie^{(i)}_{aa})(-1)^{p(b)}(E_{ab}E_{ba} - E_{aa})\right)+ (-1)^{i}\langle \Omega^{i} \vert(\sum^{n}_{i = 1}z_ie^{(i)}_{aa}) - \\ - \langle \Omega^{i} \vert\left((\sum^{n}_{i = 1}z_ie^{(i)}_{bb})(-1)^{p(a)}(E_{ba}E_{ab} - E_{bb})\right)- (-1)^i\langle \Omega^{i} \vert(\sum^{n}_{i = 1}z_ie^{(i)}_{bb}). \label{open}
          \end{gathered}
\end{equation}
One can write the projectors $\eqref{proj+},$ $\eqref{proj-}$ in the following form
\begin{equation}
\begin{gathered}
\langle \Omega^{i}\vert = \sum_{a_1 \neq \ldots \neq a_{n+m}} \langle e_{a_1}\otimes\ldots\otimes e^{(i)}_{a}\otimes\ldots\otimes e^{(j)}_{b}\otimes\ldots\otimes e_{a_{n+m})}\vert f_i(a_1,\ldots,a_{n+m}).
\end{gathered}
\end{equation}
where the functions $f_i$ are generalizations of the sign function of the ordinary permutation to the graded case. \par
Let us set $i  = 0$ and consider $p(a) = 0$, $p(b) = 1$ (the most non-trivial case) explicitly
\begin{equation}
\begin{gathered}
    \sum^{n}_{i = 1}\left(\sum_{j,\neq i}\sum_{\{a_i = a, a_j = b\}} \langle e_{a_1}\otimes\ldots\otimes e^{(i)}_{a}\otimes\ldots\otimes e^{(j)}_{b}\otimes\ldots\otimes e_{a_{n+m})}\vert f_0(a_1,\ldots,a_{n+m})\right)\times \\ \times z_i\left((-1)^{p(b)}(E_{ab}E_{ba} - E_{aa}) +1\right)  =  \\ = \sum^{n+m}_{i = 1}\Bigg(\sum_{j,\neq i}\sum_{\{a_i = a, a_j = b\}} \Big[ \langle e_{a_1}\otimes\ldots\otimes e^{(i)}_{b}\otimes\ldots\otimes e^{(j)}_{a}\otimes\ldots\otimes e_{a_{m+N}}\vert(-1)^{\sum_{k\in\{i,j\}}p(a_k)} + \\ + \langle e_{\sigma(1)}\otimes\ldots\otimes e^{(i)}_{a}\otimes\ldots\otimes e^{(j)}_{b}\otimes\ldots\otimes e_{a_{n+m}}\vert \Big]f_0(a_1,\ldots,a,\ldots,b,\ldots,a_{n+m})\Bigg) z_i,
    \end{gathered}
\end{equation}
where $\sum_{\{a_i = a, a_j = b\}} =\sum_{a_1\neq \ldots\neq a\neq \ldots\neq b\neq\ldots \neq a_{m+n}}$ and $\sum_{k\in\{i,j\}}$ means, that summation runs over all $k$, that lies in the interval strictly between $i$ and $j$. So the second line gives the following 
\begin{equation}
    \begin{gathered}
    \sum^{n+m}_{i = 1}\Bigg(\sum_{j,\neq i}\sum_{\{a_i = b, a_j = a\}} \Big[ \langle e_{a_1}\otimes\ldots\otimes e^{(i)}_{a}\otimes\ldots\otimes e^{(j)}_{b}\otimes\ldots\otimes e_{a_{m+N}}\vert(-1)^{\sum_{k\in\{i,j\}}p(a_k)} + \\ + \langle e_{\sigma(1)}\otimes\ldots\otimes e^{(i)}_{b}\otimes\ldots\otimes e^{(l)}_{a}\otimes\ldots\otimes e_{a_{n+m}}\vert \Big]f_0(a_1,\ldots,b,\ldots,a,\ldots,a_{n+m})\Bigg) z_i.
    \end{gathered}
\end{equation}
Because of $\eqref{eigenp}$ one has the following relation $$f_0(a_1,\ldots,a,\ldots,b,\ldots,a_{n+m}) = (-1)^{\sum_{k\in\{i,j\}}p(a_k)}f_0(a_1,\ldots,b,\ldots,a,\ldots,a_{n+m}).$$ From which one sees that difference $\eqref{open}$ vanishes. The other cases of parity of $a$ and $b$ as well as other value of $i$ could be considered in a similar manner.
\end{proof}
\section{Summary}
To conclude we have proved the compatibility of proposed system of Dynamical equations and its commutativity with twisted KZ equations. The structure of the formula $\eqref{sDyn}$ is very much like in \cite{FMTV}, but not quite because of fermionic roots. It will be interesting to generalize the work of \cite{EV} to the Lie superalgebra case as well as compute the monodromy of the proposed Dynamical equations. It will be addressed elsewhere \cite{E}.
\section{Appendix}  
Here we give a brief summary of some definitions of the Lie superalgebra $\mathfrak{gl}(n\vert m)$\par
Let $\mathcal{J} = \{1,\ldots n+m\}$ and let $p:\mathcal{J} \to \{0,1\}$ 
\begin{equation}
\begin{cases}
    p(a) = 0,  \, a  \leq n, \text{(bosons)} \, , \\
    p(a) = 1, \, a> n \text{(fermions)} .
     \end{cases}
 \end{equation}
 The $\mathfrak{gl}(n\vert m)$ algebra is generated by $e_{ab}$ where $a,b\in \mathcal{J}$ with the following relations 
 \begin{equation}
     e_{ab}e_{cd} - (-1)^{p(e_{ab})p(e_{cd})}e_{cd}e_{ab} = \delta_{bc} e_{ad} - (-1)^{p(e_{ab})p(e_{cd})}\delta_{da}e_{cb}, \label{definingglnm}
     \end{equation}
    where
    \begin{equation}
        p(e_{ab}) = p(a) + p(b)\, \text{mod 2}.
    \end{equation}
    The $\otimes$ product of the superalgebra representations is defined in such a way, that for operators, which has definite parity, and acting non-trivially only in the $i^{th}$ and $j^{th}$ component of the tensor product the following holds
    \begin{equation}
        A^{(i)}B^{(j)} =(-1)^{p(A)p(B)} B^{(j)}A^{(i)}. \label{superotimes}
    \end{equation}
    In the $\mathbb{C}^{n\vert m}$ there is a basis $e_{a}$ such that $e_{ab}(e_c) = \delta_{bc}e_{a}$, meaning that $e_{ab}$ are just matrix units. Let $x,y\in \mathbb{C}^{n\vert m}$ with definite $p(x)$ and $p(y)$ then the graded permutation acts as follows
    \begin{equation}
        P_{12}\left(x\otimes y \right) = (-1)^{p(x)p(y)}y\otimes x. 
    \end{equation}

\end{document}